\documentclass[11pt]{amsart}
\usepackage{amsthm}
\usepackage{amsmath,amstext,amsthm,amsfonts,amssymb,amsthm}
\usepackage{mathrsfs}
\usepackage{multicol}
\usepackage{latexsym}
\usepackage{color}
\usepackage{graphicx}
\usepackage{epsf}
\usepackage{epsfig}
\usepackage{epic}
\usepackage{graphics}
\usepackage{verbatim}
\usepackage{color}
\usepackage{hyperref}
\usepackage{bm}
\newtheorem{theorem}{Theorem}[section]
\newtheorem{lemma}[theorem]{Lemma}
\newtheorem{corollary}[theorem]{Corollary}
\theoremstyle{definition}
\newtheorem{definition}[theorem]{Definition}
\newtheorem{proposition}[theorem]{Proposition}

\theoremstyle{remark}
\newtheorem{remark}[theorem]{Remark}
\numberwithin{equation}{section}

\newcommand{\HI}{\mathfrak{H}}

\newcommand{\B}{\mathcal{B}}

\newcommand{\qu}{\mathfrak{q}}
\newcommand{\pu}{\mathfrak{p}}
\newcommand{\oqu}{\overline{\mathfrak{q}}}

\newcommand{\quat}{\mathbb H}
\newcommand{\R}{\Bbb R}
\newcommand{\Z}{\mathbb Z}
\newcommand{\mc}{\mathcal}

\newcommand{\be}{\begin{equation}}
\newcommand{\en}{\end{equation}}
\newcommand{\D}{{\mc D}}

\newcommand{\FF}{\mc F}
\newcommand{\GG}{\mc G}

\newcommand{\N}{\mathbb N}

\newcommand{\bedefin}{\begin{defi}}
	\newcommand{\findefi}{\end{defi} \medskip}
\newcommand{\betheo}{\begin{theorem}$\!\!${\bf \,\,\,}}
	\newcommand{\entheo}{\end{theorem}}
\newcommand{\enth}{\end{theorem}}
\newcommand{\becor}{\begin{cor}$\!\!${\bf .}}
	\newcommand{\encor}{\end{cor}}
\newcommand{\belem}{\begin{lem}$\!\!${\bf .}}
	\newcommand{\enlem}{\end{lem}}
\newcommand{\bea}{\begin{eqnarray}}
\newcommand{\ena}{\end{eqnarray}}
\newcommand{\beano}{\begin{eqnarray*}}
	\newcommand{\enano}{\end{eqnarray*}}
\newcommand{\bee}{\begin{enumerate}}
	\newcommand{\ene}{\end{enumerate}}
\newcommand{\bei}{\begin{itemize}}
	\newcommand{\eni}{\end{itemize}}
\newcommand{\betab}{\begin{tabular}}
	\newcommand{\entab}{\end{tabular}}

\newcommand{\Iop}{{\mathbb{I}_{V_{\mathbb{H}}^{R}}}}

\newcommand{\bfrakq}{\mbox{\boldmath $\mathfrak q$}}

\newcommand{\bfraka}{\mbox{\boldmath $\mathfrak a$}}
\newcommand{\bfrakb}{\mbox{\boldmath $\mathfrak b$}}

\newcommand{\bfrakp}{\mbox{\boldmath $\mathfrak p$}}

\newcommand{\bk}{\mathbf k}

\newcommand{\bi}{\mathbf i}
\newcommand{\bj}{\mathbf j}

\newcommand{\vr}{V_\quat^R}
\newcommand{\ur}{U_\quat^R}
\newcommand{\ra}{\text{ran}}
\newcommand{\kr}{\text{ker}}
\newcommand{\ind}{\text{ind}}

\newcommand{\Iopu}{\mathbb{I}_{\ur}}
\newcommand{\sel}{\sigma^S_{el}}
\newcommand{\ser}{\sigma^S_{er}}
\newcommand{\se}{\sigma_e^S}
\newcommand{\Wy}{\mathcal{W}}
\newcommand{\ws}{\sigma_w^S}
\newcommand{\sk}{\sigma_k^S}
\newcommand{\cK}{\mathcal{K}}
\newcommand{\cR}{\mathcal{R}}
\newcommand{\iso}{\sigma_{\text{iso}}^S}
\newcommand{\acc}{\sigma_{\text{acc}}^S}
\newcommand{\asc}{\text{asc}}
\newcommand{\dsc}{\text{dsc}}
\newcommand{\Br}{\mathfrak{Br}}
\newcommand{\Bs}{\sigma_b^S}
\begin{document}
\title[Weyl and Browder S-spectra]{Weyl and Browder S-spectra in a right quaternionic Hilbert space}
\author{B. Muraleetharan$^{\dagger}$, K. Thirulogasanthar$^{\ddagger}$}
\address{$^{\dagger}$ Department of mathematics and Statistics, University of Jaffna, Thirunelveli, Sri Lanka.}
\address{$^{\ddagger}$ Department of Computer Science and Software Engineering, Concordia University, 1455 De Maisonneuve Blvd. West, Montreal, Quebec, H3G 1M8, Canada.}
\email{bbmuraleetharan@jfn.ac.lk, santhar@gmail.com }
\date{\today}
\thanks{K. Thirulogasanthar would like to thank the FRQNT, Fonds de la Recherche  Nature et  Technologies (Quebec, Canada) for partial financial support under the grant number 2017-CO-201915. Part of this work was done while he was visiting the University of Jaffna to which he expresses his thanks for the hospitality.
}
\date{\today}
\begin{abstract}
 In this note first we study the Weyl operators and Weyl S-spectrum of a bounded right quaternionic linear operator, in the setting of the so-called S-spectrum, in a right quaternionic Hilbert space. In particular, we give a characterization for the S-spectrum in terms of the Weyl operators.  In the same space we also study the Browder operators and introduce the Browder S-spectrum.
\end{abstract}
\maketitle
\pagestyle{myheadings}
%%%%%%%%%%%%%%%%%%%%%%%%%%%%%%%%%%%%%%%%%%%%%%%%%%%%%%%%%%%%%%%%%%%%%%%%
\section{Introduction}
In the complex theory the concept of Weyl spectrum and Browder spectrum are subjects of the theory of perturbation of the spectrum, however it has found applications in operator theory and related areas \cite{Jeri, con, kub}. In the complex case, the Weyl spectrum of a bounded linear operator is the largest part of the spectrum that is invariant under compact perturbations \cite{kub, con}. We shall show that the same is true in the quaternionic Weyl S-spectrum. However, in the complex case, the Browder spectrum is not invariant under compact perturbations \cite{kub}.\\

In the complex setting, in a Hilbert space $\HI$,  for a bounded linear operator, $A$, the point spectrum or the eigenvalues of $A$ contains  isolated eigenvalues of finite algebraic and geometric multiplicities. Also these sets are important in the study of Weyl and Browder spectra \cite{kub}. In the quaternionic setting, let $\vr$ be a separable right Hilbert space,  $A$ be a bounded right linear operator, and $R_\qu(A)=A^2-2\text{Re}(\qu)A+|\qu|^2\Iop$, with $\qu\in\quat$, the set of all quaternions and $\Iop$ be the identity operator on $\vr$, be the pseudo-resolvent operator, the set of right eigenvalues of $R_\qu(A)$ coincide with the point S-spectrum (see proposition 4.5 in \cite{ghimorper}). In this regard, it will be appropriate to define and study the quaternionic isolated S-point spectrum as the quaternions which are eigenvalues of $R_\qu(A)$.\\

Due to the non-commutativity, in the quaternionic case  there are three types of  Hilbert spaces: left, right, and two-sided, depending on how vectors are multiplied by scalars. This fact can entail several problems. For example, when a Hilbert space $\mathcal H$ is one-sided (either left or right) the set of linear operators acting on it does not have a linear structure. Moreover, in a one sided quaternionic Hilbert space, given a linear operator $A$ and a quaternion $\mathfrak{q}\in\quat$, in general we have that $(\mathfrak{q} A)^{\dagger}\not=\overline{\mathfrak{q}} A^{\dagger}$ (see \cite{Mu} for details). These restrictions can severely prevent the generalization  to the quaternionic case of results valid in the complex setting. Even though most of the linear spaces are one-sided, it is possible to introduce a notion of multiplication on both sides by fixing an arbitrary Hilbert basis of $\mathcal H$.  This fact allows to have a linear structure on the set of linear operators, which is a minimal requirement to develop a full theory. Thus, the framework of this paper, is in part, is a right quaternionic Hilbert space equipped with a left multiplication, introduced by fixing a Hilbert basis.\\

In the study of Weyl and Browder S-spectra, the essential S-spectra gets involved. In defining the essential S-spectrum the structure of the so-called quaternionic Calkin algebra is used, in which the set of all bounded quaternionic right linear operators, $\B(\vr)$, should form a quaternionic two-sided Banach $C^*$-algebra with unity. This can only happen if we consider $\vr$ with a left multiplication defined on it, which is a basis dependent multiplication \cite{ghimorper}. However, regardless of which basis we choose the set $\B(\vr)$ will become a quaternionic two-sided Banach $C^*$-algebra with unity. Thus, the invariance under a basis change naturally exists.\\

As far as we know, Weyl and Browder operators and the Weyl and Browder S-spectrum have not been studied in the quaternionic setting yet. In this regard, in this note we investigate the quaternionic Weyl operators and Weyl S-spectrum and provide a characterization to the S-spectrum in terms of the weyl operators (see theorem \ref{WT1}). We also study the Browder operators to certain extent and introduce the Browder spectrum. However, in the complex case, the Browder spectrum and its characterizations depend on the so-called  Riesz idempotent which is defined in terms of the Cauchy integral formula for operators \cite{kub}. In the quaternionic setting, the Cauchy integral formula, and thereby the S-functional calculus, is known only for the slice regular functions and it is defined on an axially symmetric domain in quaternion slices. A quaternion slice is a complex plane contained in the set of all quaternions \cite{GP, ACS, CG}. In this regard, this fact severely affected our ability in studying the Browder S-spectrum in broad on the whole set of quaternions. However, one may be able to study it in an axially symmetric domain. Also in the study of quaternionic Weyl and Browder operators and S-spectra results regarding quaternionic Fredholm operators and quaternionic essential S-spectrum are involved. Materials regarding these two topics are heavily borrowed from the recent paper \cite{MT} as needed here.\\

The article is organized as follows. In section 2 we introduce the set of quaternions and quaternionic Hilbert spaces and their bases, as needed for the development of this article, which may not be familiar to a broad range of audience. In section 3 we define and investigate, as needed, right linear operators and their properties. In section 3.1 we define a basis dependent left multiplication on a right quaternionic Hilbert space. In section 3.2 we deal with the right S-spectrum, left S-spectrum, S-spectrum and its major partitions. In section 4 we recall some facts about the Fredholm operators and its index for a bounded quaternionic right linear operator from \cite{MT}. In section 5, from \cite{MT} we recall results about the essential S-spectrum as needed. We also prove certain results which are omitted from \cite{MT}. In section 6 we introduce quaternionic Weyl operators and Weyl S-spectrum. In particular we provide a characterization to the S-spectrum in terms of the quaternionic Weyl operators. In section 7 we define and study the quaternionic Browder operators and Browder S-spectrum in a limited sense, which is due to the unavailability of a Cauchy integral formula on the whole set of quaternions. Section 8 ends the manuscript with a brief conclusion.

%%%%%%%%%%%%%%%%%%%%%%%%%%%%%%%%%%%%%%%%%%%%%%%%%%%%%%%%%%%%%%%%%%%%%%
\section{Mathematical preliminaries}
In order to make the paper self-contained, we recall some facts about quaternions which may not be well-known.  For details we refer the reader to \cite{Ad,ghimorper,Vis}.
\subsection{Quaternions}
Let $\quat$ denote the field of all quaternions and $\quat^*$ the group (under quaternionic multiplication) of all invertible quaternions. A general quaternion can be written as
$$\bfrakq = q_0 + q_1 \bi + q_2 \bj + q_3 \bk, \qquad q_0 , q_1, q_2, q_3 \in \mathbb R, $$
where $\bi,\bj,\bk$ are the three quaternionic imaginary units, satisfying
$\bi^2 = \bj^2 = \bk^2 = -1$ and $\bi\bj = \bk = -\bj\bi,  \; \bj\bk = \bi = -\bk\bj,
\; \bk\bi = \bj = - \bi\bk$. The quaternionic conjugate of $\bfrakq$ is
$$ \overline{\bfrakq} = q_0 - \bi q_1 - \bj q_2 - \bk q_3 , $$
while $\vert \bfrakq \vert=(\bfrakq \overline{\bfrakq})^{1/2} $ denotes the usual norm of the quaternion $\bfrakq$.
If $\bfrakq$ is non-zero element, it has inverse
$
\bfrakq^{-1} =  \frac {\overline{\bfrakq}}{\vert \bfrakq \vert^2 }.$
Finally, the set
\begin{eqnarray*}
\mathbb{S}&=&\{I=x_1 \bi+x_2\bj+x_3\bk~\vert
~x_1,x_2,x_3\in\mathbb{R},~x_1^2+x_2^2+x_3^2=1\},
\end{eqnarray*}
contains all the elements whose square is $-1$. It is a $2$-dimensional sphere in $\mathbb H$ identified with $\mathbb R^4$.
%%%%%%%%%%%%%%%%%%%%%%%%%%%%%%%%%%%%%%%%%%%%%%%%%%%%%%%%%
\subsection{Quaternionic Hilbert spaces}
In this subsection we  discuss right quaternionic Hilbert spaces. For more details we refer the reader to \cite{Ad,ghimorper,Vis}.
\subsubsection{Right quaternionic Hilbert Space}
Let $V_{\quat}^{R}$ be a vector space under right multiplication by quaternions.  For $\phi,\psi,\omega\in V_{\quat}^{R}$ and $\bfrakq\in \quat$, the inner product
$$\langle\cdot\mid\cdot\rangle:V_{\quat}^{R}\times V_{\quat}^{R}\longrightarrow \quat$$
satisfies the following properties
\begin{enumerate}
	\item[(i)]
	$\overline{\langle \phi\mid \psi\rangle}=\langle \psi\mid \phi\rangle$
	\item[(ii)]
	$\|\phi\|^{2}=\langle \phi\mid \phi\rangle>0$ unless $\phi=0$, a real norm
	\item[(iii)]
	$\langle \phi\mid \psi+\omega\rangle=\langle \phi\mid \psi\rangle+\langle \phi\mid \omega\rangle$
	\item[(iv)]
	$\langle \phi\mid \psi\bfrakq\rangle=\langle \phi\mid \psi\rangle\bfrakq$
	\item[(v)]
	$\langle \phi\bfrakq\mid \psi\rangle=\overline{\bfrakq}\langle \phi\mid \psi\rangle$
\end{enumerate}
where $\overline{\bfrakq}$ stands for the quaternionic conjugate. It is always assumed that the
space $V_{\quat}^{R}$ is complete under the norm given above and separable. Then,  together with $\langle\cdot\mid\cdot\rangle$ this defines a right quaternionic Hilbert space. Quaternionic Hilbert spaces share many of the standard properties of complex Hilbert spaces.

The next two Propositions can be established following the proof of their complex counterparts, see e.g. \cite{ghimorper,Vis}.
\begin{proposition}\label{P1}
Let $\mathcal{O}=\{\varphi_{k}\,\mid\,k\in N\}$
be an orthonormal subset of $V_{\quat}^{R}$, where $N$ is a countable index set. Then following conditions are pairwise equivalent:
\begin{itemize}
\item [(a)] The closure of the linear combinations of elements in $\mathcal O$ with coefficients on the right is $V_{\quat}^{R}$.
\item [(b)] For every $\phi,\psi\in V_{\quat}^{R}$, the series $\sum_{k\in N}\langle\phi\mid\varphi_{k}\rangle\langle\varphi_{k}\mid\psi\rangle$ converges absolutely and it holds:
$$\langle\phi\mid\psi\rangle=\sum_{k\in N}\langle\phi\mid\varphi_{k}\rangle\langle\varphi_{k}\mid\psi\rangle.$$
\item [(c)] For every  $\phi\in V_{\quat}^{R}$, it holds:
$$\|\phi\|^{2}=\sum_{k\in N}\mid\langle\varphi_{k}\mid\phi\rangle\mid^{2}.$$
\item [(d)] $\mathcal{O}^{\bot}=\{0\}$.
\end{itemize}
\end{proposition}
\begin{definition}
The set $\mathcal{O}$ as in Proposition \ref{P1} is called a {\em Hilbert basis} of $V_{\quat}^{R}$.
\end{definition}
\begin{proposition}\label{P2}
Every quaternionic separable Hilbert space $V_{\quat}^{R}$ has a Hilbert basis. All the Hilbert bases of $V_{\quat}^{R}$ have the same cardinality.

Furthermore, if $\mathcal{O}$ is a Hilbert basis of $V_{\quat}^{R}$, then every  $\phi\in V_{\quat}^{R}$ can be uniquely decomposed as follows:
$$\phi=\sum_{k\in N}\varphi_{k}\langle\varphi_{k}\mid\phi\rangle,$$
where the series $\sum_{k\in N}\varphi_k\langle\varphi_{k}\mid\phi\rangle$ converges absolutely in $V_{\quat}^{R}$.
\end{proposition}

It should be noted that once a Hilbert basis is fixed, every left (resp. right) quaternionic Hilbert space also becomes a right (resp. left) quaternionic Hilbert space \cite{ghimorper,Vis}. See next section 3.2 for more details.

The field of quaternions $\quat$ itself can be turned into a left quaternionic Hilbert space by defining the inner product $\langle \bfrakq \mid \bfrakq^\prime \rangle = \bfrakq \overline{\bfrakq^{\prime}}$ or into a right quaternionic Hilbert space with  $\langle \bfrakq \mid \bfrakq^\prime \rangle = \overline{\bfrakq}\bfrakq^\prime$.
%%%%%%%%%%%%%%%%%%%%%%%%%%%%%%%%%%%%%%%%%%%%%%%%%%%%%%%%%%%%%%%%
%%%%%%%%%%%%%%%%%%%%%%%%%%%%%%%%%%%%%%%%%%%%%%%%%%%%%%%%%%%%%%%%%%%%%%%%%%%%%%%%%%
\section{Right quaternionic linear  operators and some basic properties}
In this section we shall define right  $\quat$-linear operators and recall some basic properties. Most of them are very well known. In this manuscript, we follow the notations of \cite{AC} and \cite{ghimorper}.
\begin{definition}
A mapping $A:\D(A)\subseteq V_{\quat}^R \longrightarrow V_{\quat}^R$, where $\D(A)$ stands for the domain of $A$, is said to be right $\quat$-linear operator or, for simplicity, right linear operator, if
$$A(\phi\bfraka+\psi\bfrakb)=(A\phi)\bfraka+(A\psi)\bfrakb,~~\mbox{~if~}~~\phi,\,\psi\in \D(A)~~\mbox{~and~}~~\bfraka,\bfrakb\in\quat.$$
\end{definition}
The set of all right linear operators will be denoted by $\mathcal{L}(V_{\quat}^{R})$ and the identity linear operator on $V_{\quat}^{R}$ will be denoted by $\Iop$. For a given $A\in \mathcal{L}(V_{\quat}^{R})$, the range and the kernel will be
\begin{eqnarray*}
\text{ran}(A)&=&\{\psi \in V_{\quat}^{R}~|~A\phi =\psi \quad\text{for}~~\phi \in\D(A)\}\\
\ker(A)&=&\{\phi \in\D(A)~|~A\phi =0\}.
\end{eqnarray*}
We call an operator $A\in \mathcal{L}(V_{\quat}^{R})$ bounded if
\begin{equation}\label{PE1}
\|A\|=\sup_{\|\phi \|=1}\|A\phi \|<\infty,
\end{equation}
or equivalently, there exist $K\geq 0$ such that $\|A\phi \|\leq K\|\phi \|$ for all $\phi \in\D(A)$. The set of all bounded right linear operators will be denoted by $\B(V_{\quat}^{R})$. Set of all  invertible bounded right linear operators will be denoted by $\mathcal{G} (V_{\quat}^{R})$. We also denote for a set $\Delta\subseteq\quat$, $\Delta^*=\{\oqu~|~\qu\in\Delta\}$.
\\
Assume that $V_{\quat}^{R}$ is a right quaternionic Hilbert space, $A$ is a right linear operator acting on it.
Then, there exists a unique linear operator $A^{\dagger}$ such that
\begin{equation}\label{Ad1}
\langle \psi \mid A\phi \rangle=\langle A^{\dagger} \psi \mid\phi \rangle;\quad\text{for all}~~~\phi \in \D (A), \psi\in\D(A^\dagger),
\end{equation}
where the domain $\D(A^\dagger)$ of $A^\dagger$ is defined by
$$
\D(A^\dagger)=\{\psi\in V_{\quat}^{R}\ |\ \exists \varphi\ {\rm such\ that\ } \langle \psi \mid A\phi \rangle=\langle \varphi \mid\phi \rangle\}.$$
\begin{proposition}\label{KR}\cite{ghimorper}
If $A\in\B(\vr)$ is normal then $\kr(A)=\kr(A^\dagger)$.
\end{proposition}
\begin{proposition}\label{IP30}\cite{MT, BT}
Let $A\in\B(\vr, \ur)$ then
$$(a)~~\ra(A)^\perp=\kr(A^\dagger).\quad(b)~~\kr(A)=\ra(A^\dagger)^\perp.$$
\end{proposition}
\begin{definition}\label{PD1}
Let $\vr$ and $\ur$ be right quaternionic  Hilbert spaces. A bounded operator $K:\vr\longrightarrow\ur$ is compact if $K$ maps bounded sets into precompact sets. That is, $\overline{K(U)}$ is compact in $\ur$, where $U=\{\phi\in\vr~|~\|\phi\|<1\}$. Equivalently, for all bounded sequences $\{\phi_n\}_{n=1}^\infty$ in $\vr$ the sequence $\{K\phi_n\}_{n=0}^\infty$ has a convergence subsequence in $\ur$.
\end{definition}
We denote the set of all compact operators from $\vr$ to $\ur$ by $\B_0(\vr,\ur)$ and the compact operators from $\vr$ from $\vr$ will be denoted by $\B_0(\vr)$.
\begin{definition}\label{PD2}
An operator $K:\vr\longrightarrow\ur$ is said to be of finite rank if $\text{ran}(K)\subseteq\ur$ is finite dimensional.
\end{definition}
\begin{proposition}\label{PD22}\cite{MT}
If $A\in\B(\vr,\ur)$ is of finite rank, then $A$ is compact.
\end{proposition} 
\begin{proposition}\label{PP02}\cite{MT}
Let $A\in\B(\vr,\ur)$ be a finite rank operator, then $A^\dagger\in\B(\ur,\vr)$ is a finite rank operator and $\dim(\ra(A))=\dim(\ra(A^\dagger))$.
\end{proposition}
\begin{definition}\label{PD3}
Let $M\subset\vr$ be a closed subspace, then $\text{codim}(M)=\dim(\vr/M)$.
\end{definition}
\begin{definition}\label{PD4}
Let $A:\vr\longrightarrow\ur$ be a bounded operator, then $\text{coker}(A):=\ur/\text{ran}(A)$ and $\dim(\text{coker}(A))=\dim(\ur)-\dim(\text{ran}(A)).$
\end{definition}
\begin{proposition}\label{PC1}\cite{MT}
A bounded operator $K:\vr\longrightarrow\ur$ is compact if and only if there exists finite rank operators $K_n:\vr\longrightarrow\ur$ such that $\|K-K_n\|\longrightarrow 0$ as $n\longrightarrow 0$.
\end{proposition}
\begin{corollary}\label{CD1}\cite{MT}
A bounded operator $K:\vr\longrightarrow\ur$ is compact then so is $K^\dagger$.
\end{corollary}
\begin{proposition}\label{PP2}\cite{MT}
Let $A\in\B(\vr)$ and $K$ be a compact operator on $\vr$, then $AK$ and $KA$ are compact operators.
\end{proposition}
%%%%%%%%%%%%%%%
\begin{definition}\label{ID1}
Let $A\in\B(\vr)$. A closed subspace $M\subseteq\vr$ is said to be invariant under $A$ if $A(M)\subseteq M$, where $A(M)=\{A\phi~|~\phi\in M\}$.
\end{definition}
\subsection{Left Scalar Multiplications on $V_{\quat}^{R}$.}
We shall extract the definition and some properties of left scalar multiples of vectors on $V_{\quat}^R$ from \cite{ghimorper} as needed for the development of the manuscript. The left scalar multiple of vectors on a right quaternionic Hilbert space is an extremely non-canonical operation associated with a choice of preferred Hilbert basis. From the Proposition \ref{P2}, $V_{\quat}^{R}$ has a Hilbert basis
\begin{equation}\label{b1}
\mathcal{O}=\{\varphi_{k}\,\mid\,k\in N\},
\end{equation}
where $N$ is a countable index set.
The left scalar multiplication on $V_{\quat}^{R}$ induced by $\mathcal{O}$ is defined as the map $\quat\times V_{\quat}^{R}\ni(\bfrakq,\phi)\longmapsto \bfrakq\phi\in V_{\quat}^{R}$ given by
\begin{equation}\label{LPro}
\bfrakq\phi:=\sum_{k\in N}\varphi_{k}\bfrakq\langle \varphi_{k}\mid \phi\rangle,
\end{equation}
for all $(\bfrakq,\phi)\in\quat\times V_{\quat}^{R}$.
\begin{proposition}\cite{ghimorper}\label{lft_mul}
The left product defined in equation \ref{LPro} satisfies the following properties. For every $\phi,\psi\in V_{\quat}^{R}$ and $\bfrakp,\bfrakq\in\quat$,
\begin{itemize}
\item[(a)] $\bfrakq(\phi+\psi)=\bfrakq\phi+\bfrakq\psi$ and $\bfrakq(\phi\bfrakp)=(\bfrakq\phi)\bfrakp$.
\item[(b)] $\|\bfrakq\phi\|=|\bfrakq|\|\phi\|$.
\item[(c)] $\bfrakq(\bfrakp\phi)=(\bfrakq\bfrakp)\phi$.
\item[(d)] $\langle\overline{\bfrakq}\phi\mid\psi\rangle
=\langle\phi\mid\bfrakq\psi\rangle$.
\item[(e)] $r\phi=\phi r$, for all $r\in \mathbb{R}$.
\item[(f)] $\bfrakq\varphi_{k}=\varphi_{k}\bfrakq$, for all $k\in N$.
\end{itemize}
\end{proposition}
Furthermore, the quaternionic left scalar multiplication of linear operators is also defined in \cite{Fab1}, \cite{ghimorper}. For any fixed $\bfrakq\in\quat$ and a given right linear operator $A:\D(A)\longrightarrow V_{\quat}^R$, the left scalar multiplication of $A$ is defined as a map $\bfrakq A:\D(A)\longrightarrow V_{\quat}^R$ by the setting
\begin{equation}\label{lft_mul-op}
(\bfrakq A)\phi:=\bfrakq (A\phi)=\sum_{k\in N}\varphi_{k}\bfrakq\langle \varphi_{k}\mid A\phi\rangle,
\end{equation}
for all $\phi\in D(A)$. It is straightforward that $\bfrakq A$ is a right linear operator. If $\bfrakq\phi\in \D(A)$, for all $\phi\in \D(A)$, one can define right scalar multiplication of the right linear operator $A:\D(A)\longrightarrow V_{\quat}^R$ as a map $ A\bfrakq:\D(A)\longrightarrow V_{\quat}^R$ by the setting
\begin{equation}\label{rgt_mul-op}
(A\bfrakq )\phi:=A(\bfrakq \phi),
\end{equation}
for all $\phi\in D(A)$. It is also a right linear operator. One can easily obtain that, if $\bfrakq\phi\in \D(A)$, for all $\phi\in \D(A)$ and $\D(A)$ is dense in $V_{\quat}^R$, then
\begin{equation}\label{sc_mul_aj-op}
(\bfrakq A)^{\dagger}=A^{\dagger}\overline{\bfrakq}~\mbox{~and~}~
(A\bfrakq)^{\dagger}=\overline{\bfrakq}A^{\dagger}.
\end{equation}

%%%%%%%%%%%%%%%%%%%%%%%%%%%%%%%%%%%%%%%%%%%%%%%
\subsection{S-Spectrum}
For a given right linear operator $A:\D(A)\subseteq V_{\quat}^R\longrightarrow V_{\quat}^R$ and $\bfrakq\in\quat$, we define the operator $R_{\bfrakq}(A):\D(A^{2})\longrightarrow\quat$ by  $$R_{\bfrakq}(A)=A^{2}-2\text{Re}(\bfrakq)A+|\bfrakq|^{2}\Iop,$$
where $\bfrakq=q_{0}+\bi q_1 + \bj q_2 + \bk q_3$ is a quaternion, $\text{Re}(\bfrakq)=q_{0}$  and $|\bfrakq|^{2}=q_{0}^{2}+q_{1}^{2}+q_{2}^{2}+q_{3}^{2}.$\\
In the literature, the operator is called pseudo-resolvent since it is not the resolvent operator of $A$ but it is the one related to the notion of spectrum as we shall see in the next definition. For more information, on the notion od $S$-spectrum the reader may consult e.g. \cite{Fab, Fab1, NFC}, and  \cite{ghimorper}.
\begin{definition}
Let $A:\D(A)\subseteq V_{\quat}^R\longrightarrow V_{\quat}^R$ be a right linear operator. The {\em $S$-resolvent set} (also called \textit{spherical resolvent} set) of $A$ is the set $\rho_{S}(A)\,(\subset\quat)$ such that the three following conditions hold true:
\begin{itemize}
\item[(a)] $\ker(R_{\bfrakq}(A))=\{0\}$.
\item[(b)] $\text{ran}(R_{\bfrakq}(A))$ is dense in $V_{\quat}^{R}$.
\item[(c)] $R_{\bfrakq}(A)^{-1}:\text{ran}(R_{\bfrakq}(A))\longrightarrow\D(A^{2})$ is bounded.
\end{itemize}
The \textit{$S$-spectrum} (also called \textit{spherical spectrum}) $\sigma_{S}(A)$ of $A$ is defined by setting $\sigma_{S}(A):=\quat\smallsetminus\rho_{S}(A)$. For a bounded linear operator $A$ we can write the resolvent set as
\begin{eqnarray*}
\rho_S(A)&=& \{\qu\in\quat~|~R_\qu(A)\in\mathcal{G}(V_{\quat}^R)\}\\
&=&\{\qu\in\quat~|~R_\qu(A)~\text{has an inverse in}~\B(V_{\quat}^R)\}\\
&=&\{\qu\in\quat~|~\text{ker}(R_\qu(A))=\{0\}\quad\text{and}\quad \text{ran}(R_\qu(A))=V_\quat^R\}
\end{eqnarray*}
and the spectrum can be written as
\begin{eqnarray*}
\sigma_S(A)&=&\quat\setminus\rho_S(A)\\
&=&\{\qu\in\quat~|~R_\qu(A)~\text{has no inverse in}~\B(V_{\quat}^R)\}\\
&=&\{\qu\in\quat~|~\text{ker}(R_\qu(A))\not=\{0\}\quad\text{or}\quad \text{ran}(R_\qu(A))\not=V_\quat^R\}
\end{eqnarray*}
The right $S$-spectrum $\sigma_r^S(A)$ and the left $S$-spectrum $\sigma_l^S(A)$ are defined respectively as
\begin{eqnarray*}\
\sigma_r^S(A)&=&\{\qu\in\quat~|~R_\qu(A)~~\text{in not right invertible in}~~\B(V_\quat^R)~\}\\
\sigma_l^S(A)&=&\{\qu\in\quat~|~R_\qu(A)~~\text{in not left invertible in}~~\B(V_\quat^R)~\}.
\end{eqnarray*}
The spectrum $\sigma_S(A)$ decomposes into three disjoint subsets as follows:
\begin{itemize}
\item[(i)] the \textit{spherical point spectrum} of $A$: $$\sigma_{pS}(A):=\{\bfrakq\in\quat~\mid~\ker(R_{\bfrakq}(A))\ne\{0\}\}.$$
\item[(ii)] the \textit{spherical residual spectrum} of $A$: $$\sigma_{rS}(A):=\{\bfrakq\in\quat~\mid~\ker(R_{\bfrakq}(A))=\{0\},\overline{\text{ran}(R_{\bfrakq}(A))}\ne V_{\quat}^{R}~\}.$$
\item[(iii)] the \textit{spherical continuous spectrum} of $A$: $$\sigma_{cS}(A):=\{\bfrakq\in\quat~\mid~\ker(R_{\bfrakq}(A))=\{0\},\overline{\text{ran}(R_{\bfrakq}(A))}= V_{\quat}^{R}, R_{\bfrakq}(A)^{-1}\notin\B(V_{\quat}^{R}) ~\}.$$
\end{itemize}
If $A\phi=\phi\bfrakq$ for some $\bfrakq\in\quat$ and $\phi\in V_{\quat}^{R}\smallsetminus\{0\}$, then $\phi$ is called an \textit{eigenvector of $A$ with right eigenvalue} $\bfrakq$. The set of right eigenvalues coincides with the point $S$-spectrum, see \cite{ghimorper}, Proposition 4.5.
\end{definition}
\begin{proposition}\cite{Fab2, ghimorper}\label{PP1}
For $A\in\B(\vr)$, the resolvent set $\rho_S(A)$ is a non-empty open set and the spectrum $\sigma_S(A)$ is a non-empty compact set.
\end{proposition}
\begin{proposition}\cite{CG}\label{EP}
Let $A\in\B(\vr)$ and let $\pu=p_0+p_1I\in p_0+p_1\mathbb{S}\subseteq\quat\setminus\R$ be an $S-$eigenvalue of $A$. Then all the elements of the sphere $[\pu]=p_0+p_1\mathbb{S}$ are eigenvalues of $A$.
\end{proposition}
%%%%%%%%%%%%%%%%%%%%
\begin{proposition}\cite{MT}\label{SP1}
Let $A\in\B(\vr)$.
\begin{eqnarray}
\sigma_l^S(A)&=&\{\qu\in\quat~~|~~\ra(R_\qu(A))~~\text{is closed or}~~\kr(R_\qu(A))\not=\{0\}\}\label{SE1}.\\
\sigma_r^S(A)&=&\{\qu\in\quat~~|~~\ra(R_\qu(A))~~\text{is closed or}~~\kr(R_{\oqu}(A^\dagger))\not=\{0\}\}\label{SE2}.
\end{eqnarray}
\end{proposition}
%%%%%%%%%%%%%%%%%%%%%%%%%%%%%%%%%%%%%%%%%%%%%%%%%%%%
\section{Fredholm operators in the quaternionic setting}
In order to study the Weyl and Browder operators and Weyl and Browder S-spectra we need some results regarding the Fredholm operators. We borrow the materials of this section from \cite{MT} as needed for the development of the manuscript. For an enhanced explanation we refer the reader to \cite{MT}. In this regard let $\vr$ and $\ur$ be two separable right quaternionic Hilbert spaces.
\begin{definition}\label{FD1}
A Fredholm operator is an operator $A\in\B(\vr,\ur)$ such that $\text{ker}(A)$ and $\text{coker}(A)=\ur/\text{ran}(A)$ are finite dimensional. The dimension of the cokernel is called the codimension, and it is denoted by $\text{codim}(A).$
\end{definition}
%%%%%%%%%
\begin{proposition} \cite{MT}\label{FP1}
If $A\in\B(\vr,\ur)$ is a Fredholm operator, then $\text{ran}(A)$ is closed.
\end{proposition}
%%%%%%%
\begin{definition}\label{FD2}
Let $A\in\B(\vr,\ur)$ be a Fredholm operator. Then the index of $A$ is the integer, $\text{ind}(A)=\dim(\text{ker}(A))-\dim(\text{coker}(A)).$
\end{definition}
\begin{remark}\label{FR1}
Since $\ra(A)$ is closed, we have $\ur=\ra(A)\oplus\ra(A)^\perp=\ra(A)\oplus\kr(A^\dagger).$
Therefore, $\text{coker}(A)=\ur/\ra(A)\cong\kr(A^\dagger).$ Thus,
$$\text{ind}(A)=\dim(\kr(A))-\dim(\kr(A^\dagger)).$$
\end{remark}
%%%%%%%%%%%%%%%%
\begin{theorem} \cite{MT}\label{FT1}
	Let $A\in\B(\vr,\ur)$ be bijective, and let $K\in\B_0(\vr,\ur)$ be compact. Then $A+K$ is a Fredholm operator.
\end{theorem}
%%%%%%%%%%%%%%%%%%%%%%%%%%%%%
\begin{proposition} \cite{MT}\label{FP2}
If $A\in\B(\vr,\ur)$ is Fredholm then $A^\dagger\in\B(\ur,\vr)$ is Fredholm.
\end{proposition}

%%%%%%%%%%%%
\begin{theorem} \cite{MT}\label{FT3}
$A\in\B(\vr,\ur)$ is Fredholm if and only if there exist $S_1,S_2\in\B(\ur,\vr)$ and compact operators $K_1$ and $K_2$, on $\vr$ and $\ur$ respectively, such that
$$S_1A=\mathbb{I}_{\vr}+K_1\quad\text{and}\quad AS_2=\mathbb{I}_{\ur}+K_2.$$
\end{theorem}
%%%%%%%%%%%%%%%%%%
\begin{remark} \cite{MT}\label{FR4}
Let $A\in\B(\vr,\ur),$ then
\begin{enumerate}
\item [(a)] $A$ is said to be left semi-Fredholm if there exists $B\in\B(\ur,\vr)$ and a compact operator $K_1$ on $\vr$ such that $BA=\Iop+K_1$. The set of all left semi-Fredholm operators are denoted by $\mathcal{F}_l(\vr,\ur)$ \cite{con}.
\item [(b)] $A$ is said to be right semi-Fredholm if there exists $B\in\B(\ur,\vr)$ and a compact operator $K_2$ on $\ur$ such that $AB=\Iopu+K_2$. The set of all right semi-Fredholm operators are denoted by $\mathcal{F}_r(\vr,\ur)$ \cite{con}.
\item[(c)] By theorem \ref{FT3}, the set of all Fredholm operators, $\mathcal{F}(\vr,\ur)=\mathcal{F}_l(\vr,\ur)\cap \mathcal{F}_r(\vr,\ur)$.
\item[(d)] From theorem \ref{FT3} it is also clear that every invertible right linear operator is Fredholm.
\item[(e)]Let $\mathcal{S}\FF(\vr)=\FF_l(\vr)\cup\FF_r(\vr)$. From theorem \ref{FT3} and corollary \ref{CD1}, we have
\begin{eqnarray*}
A\in\FF_l(\vr)&\Leftrightarrow& A^\dagger\in\FF_r(\vr)\\
A\in\mathcal{S}\FF(\vr)&\Leftrightarrow& A^\dagger\in\mathcal{S}\FF(\vr)\\
A\in\FF(\vr)&\Leftrightarrow& A^\dagger\in\FF(\vr).
\end{eqnarray*}
\end{enumerate}
\end{remark}
%%%%%%%%%%
\begin{theorem} \cite{MT}\label{FT2}
Let $\vr, \ur$ and $W_\quat^R$ be right quaternionic Hilbert spaces. If $A_1\in\B(\vr,\ur)$ and $A_2\in\B(\ur,W_\quat^R)$ are two Fredholm operators, then $A_2A_1\in\B(\vr, W_\quat^R)$ is also a Fredholm operator, and it satisfies $\ind(A_2A_1)=\ind(A_1)+\ind(A_2)$.
\end{theorem}
%%%%%%%%%%%%%%%%%
\begin{lemma} \cite{MT}\label{FL2}
Let $F\in\B(\vr)$ be a finite rank operator, then $\ind(\Iop+F)=0.$
\end{lemma}
%%%%%%%%%%%%%%%%%%%%%
\begin{theorem} \cite{MT}\label{FT4}
Let $A\in\B(\vr,\ur)$ be a Fredholm operator, then for any compact operator $K\in\B(\vr,\ur)$, $A+K$ is a Fredholm operator and $\ind(A+K)=\ind(A)$.
\end{theorem}
%%%%%%%%%%%%%%%%%%%%
\begin{corollary} \cite{MT}\label{CI1} Every invertible operator $A\in\B(\vr)$ is Fredholm and $\ind(A)=0$.
\end{corollary}
%%%%%%%%%%%%%%%%%%%%%%%%%%%%%%%%%
\begin{corollary} \cite{MT}\label{CN}
Let $n$ be a non-negative integer. If $A\in\FF(\vr)$ then $A^n\in\FF(\vr)$ and $\ind(A^n)=n~\ind(A)$.
\end{corollary}
%%%%%%%%%%%%%%%%%%%%%%%%
\begin{theorem} \cite{MT}\label{FT6}
An operator $A\in\B(\vr)$ is left semi-Fredholm if and only if $\ra(A)$ is closed and $\kr(A)$ is finite dimensional. Hence
\begin{eqnarray}
\FF_l(\vr)&=&\{A\in\B(\vr)~~|~~\ra(A)~\text{ is closed and}~ \dim(\kr(A))<\infty\}\\
\FF_r(\vr)&=&\{A\in\B(\vr)~~|~~\ra(A)~\text{ is closed and}~ \dim(\kr(A^\dagger))<\infty\}
\end{eqnarray}
\end{theorem}
%%%%%%%%%%%%%%%%%%%%%
\begin{remark} \cite{MT}\label{ER2}
Let $A\in\B(\vr)$.
\begin{enumerate}
\item [(a)] The so-called Weyl operators are Fredholm operators on $\vr$ with null index. That is, the set of all Weyl operators,
$$\Wy(\vr)=\{A\in\FF(\vr)~~|~~\ind(A)=0\}.$$
\item[(b)] Since, by remark \ref{FR4} (e), $A\in\FF(\vr)\Leftrightarrow A^\dagger\in\FF(\vr)$ and $\ind(A)=-\ind(A^\dagger)$, $A\in\Wy(\vr)\Leftrightarrow A^\dagger\in\Wy(\vr)$.
\item[(c)] By theorem \ref{FT1} and lemma \ref{FL2}, if $F$ is a finite rank operator, then $\Iop+F\in\Wy(\vr)$. 
\item[(d)]By theorem \ref{FT2}, $A,B\in\Wy(\vr)\Rightarrow AB\in\Wy(\vr)$.
\item[(e)]By theorem \ref{FT4}, $A\in\Wy(\vr), K\in\B_0(\vr)\Rightarrow A+K\in\Wy(\vr).$

\item[(f)] By corollary \ref{CI1}, $A\in\B(\vr)$ is invertible, then $A\in\Wy(\vr)$
\item[(g)]Suppose $\dim(\vr)<\infty$, then
$\ind(A)=0$ for any $A\in\B(\vr)$.
Therefore, every operator in $\B(\vr)$ is a Fredholm operator with index zero. In this case, $\Wy(\vr)=\B(\vr)$.
\end{enumerate}
\end{remark}
%%%%%%%%%%%%%%%%%%%%%%%%%%%%%%%%%%%%%%%%%%%%%%%%
\section{essential S-spectrum}
 Most part of this section is borrowed from \cite{MT} as needed here. For details we refer the reader to \cite{MT}. We also give proofs to some new results which are omitted in \cite{MT}.
\begin{theorem}\cite{ghimorper}\label{ET0}
Let $\vr$ be a right quaternionic Hilbert space equipped with a left scalar multiplication. Then the set $\B(\vr)$ equipped with the point-wise sum, with the left and right scalar multiplications defined in equations \ref{lft_mul-op} and \ref{rgt_mul-op}, with the composition as product, with the adjunction $A\longrightarrow A^\dagger$, as in equation \ref{Ad1}, as $^*-$ involution and with the norm defined in equation \ref{PE1}, is a quaternionic two-sided Banach $C^*$-algebra with unity $\Iop$.
\end{theorem}
\begin{remark}\label{ER0}
In the above theorem, if the left scalar multiplication is left out on $\vr$, then $\B(\vr)$ becomes a real Banach $C^*$-algebra with unity $\Iop$.
\end{remark}
%%%%%%%%%%%%%
\begin{theorem}\cite{Fa}\label{ET1}
The set of all compact operators, $\B_0(\vr)$ is a closed biideal of $\B(\vr)$ and is closed under adjunction.
\end{theorem}
On the quotient space $B(\vr)/B_0(\vr)$ the coset of $A\in\B(\vr)$ is
$$[A]=\{S\in\B(\vr)~|~S=A+K~~~~\text{for some}~~~K\in\B_0(\vr)\}=A+\B_0(\vr).$$
On the quotient space define the product
$$[A][B]=[AB].$$
Since $\B_0(\vr)$ is a closed subspace of $\B(\vr)$, with the above product, $\B(\vr)/\B_0(\vr)$ is a unital Banach algebra with unit $[\Iop]$. We call this algebra the quaternionic Calkin algebra. Define the natural quotient map
$$\pi:\B(\vr)\longrightarrow \B(\vr)/\B_0(\vr)\quad\text{by}\quad \pi(A)=[A]=A+\B_0(\vr).$$
Note that $[0]=\B_0(\vr)$ and hence
$$\kr(\pi)=\{A\in\B(\vr)~~|~~\pi(A)=[0]\}=\B_0(\vr).$$
Since $\B_0(\vr)$ is an ideal of $\B(\vr)$, for $A,B\in\B(\vr)$, we have
\begin{enumerate}
\item[(a)]$\pi(A+B)=(A+B)+\B_0(\vr)=(A+\B_0(\vr))+(B+\B_0(\vr))=\pi(A)+\pi(B)$.
\item[(b)]$\pi(AB)=AB+\B_0(\vr)=(A+\B_0(\vr))(B+\B_0(\vr))=\pi(A)\pi(B).$
\item[(c)]$\pi(\Iop)=[\Iop]$.
\end{enumerate}
Hence $\pi$ is a unital homomorphism. The norm on $\B(\vr)/\B_0(\vr)$ is given by
$$\|[A]\|=\inf_{K\in\B_0(\vr)}\|A+K\|\leq\|A\|.$$
Therefore $\pi$ is a contraction.
%%%%%%%%%%%%%%%%%%%%%%%
\begin{definition}\label{ED1}
The essential $S$-spectrum (or the Calkin $S$-spectrum) $\se(A)$ of $A\in\B(\vr)$ is the $S$-spectrum of $\pi(A)$ in the unital Banach algebra $B(\vr)/\B_0(\vr)$. That is, $$\se(A)=\sigma_S(\pi(A)).$$ Similarly, the left essential $S$-spectrum $\sel(A)$ and the right essential $S$-spectrum $\ser(A)$ are the left and right $S$-spectrum of $\pi(A)$ respectively. That is,
$$\sel(A)=\sigma_l^S(\pi(A))\quad\text{and}\quad \ser(A)=\sigma_r^S(\pi(A))$$ in  $B(\vr)/\B_0(\vr)$.\\
Clearly, by definition, $\se(A)=\sel(A)\cup\ser(A)$ and $\se(A)$ is a compact subset of $\quat$.
\end{definition}
%%%%%%%%%%%%%%%%%%%%%%%
\begin{proposition} \cite{MT}\label{EP1}
Let $A\in\B(\vr)$, then
\begin{eqnarray}
\sel(A)&=&\{\qu\in\quat~|~R_\qu(A)\in\B(\vr)\setminus\mathcal{F}_l(\vr)\}\label{EE1}\\
\ser(A)&=&\{\qu\in\quat~|~R_\qu(A)\in\B(\vr)\setminus\mathcal{F}_r(\vr)\}\label{EE2}
\end{eqnarray}
\end{proposition}
%%%%%%%%%%%%%%%
\begin{corollary} \cite{MT}(Atkinson theorem)\label{EC2}
Let $A\in\B(\vr)$, then
\begin{equation}\label{EE3}
\se(A)=\{\qu\in\quat~~~|~~~R_\qu(A)\in\B(\vr)\setminus\mathcal{F}(\vr)\}.
\end{equation}
\end{corollary}
%%%%%%%%%%%%%%%%%
\begin{proposition}\label{EP4}
For $A\in\B(\vr)$, $\se(A)\not=\emptyset$ if and only if $\dim(\vr)=\infty.$
\end{proposition}
%%%%%%%%%%%%%%
\begin{proposition} \cite{MT}\label{EP5}
For every $A\in\B(\vr)$ and $K\in\B_0(\vr)$, we have $\se(A+K)=\se(A)$. In the same way, $\sel(A+K)=\sel(A)$ and $\ser(A+K)=\ser(A)$.
\end{proposition}
%%%%%%%%%%%%%%
\begin{definition}\label{ED2}
Let $A\in\B(\vr)$ and $k\in\Z\setminus\{0\}$. Define,
$$\sigma_k^S(A)=\{\qu\in\quat~~|~~R_\qu(A)\in\FF(\vr)\quad\text{and}\quad \ind(R_\qu(A))=k\}.$$
Also
$$\sigma_0^S=\{\qu\in\sigma_S(A)~~|~~R_\qu(A)\in\Wy(\vr)\}.$$
\end{definition}
%%%%%%%%%%%%%%%%%
\begin{proposition} \cite{MT}\label{EP7}
Let $A\in\B(\vr)$, then $\displaystyle\sigma_S(A)=\se(A)\cup\bigcup_{k\in\Z}\sigma_k^S(A).$
\end{proposition}
%%%%%%%%%%%%%%%%%%%%%%%%%%%%%%%%%%%%%%%%%%%%%%%
\begin{definition}\label{Din} For $A\in\B(\vr)$, we define
$$\sigma_{+\infty}^S(A)=\{\qu\in\quat~|~R_\qu(A)\in\mathcal{SF}(\vr)\quad\text{and}\quad \ind(R_\qu(A))=+\infty\},$$
$$\sigma_{-\infty}^S(A)=\{\qu\in\quat~|~R_\qu(A)\in\mathcal{SF}(\vr)\quad\text{and}\quad \ind(R_\qu(A))=-\infty\}.$$
\end{definition}
%%%%%%%%%%%%%%%%%%%
\begin{proposition}\label{Pin}
For $A\in\B(\vr)$ we have $$\sigma_{+\infty}^S(A)\cup\sigma_{-\infty}^S(A)=\{\qu\in\sigma_S(A)~~|~~R_\qu(A)\in\mathcal{SF}(\vr)\setminus\FF(\vr)\}.$$
\end{proposition}
\begin{proof}
We have
\begin{eqnarray*}
\sigma_{+\infty}^S(A)&=&\{\qu\in\quat~|~R_\qu(A)\in\mathcal{SF}(\vr)\quad\text{and}\quad \ind(R_\qu(A))=+\infty\}\\
&=&\{\qu\in\quat~|~R_\qu(A)\in\mathcal{SF}(\vr)\quad\text{and}\quad \dim(\kr(R_\qu(A)))=+\infty\}\\
&=&\{\qu\in\quat~|~R_\qu(A)\in\mathcal{F}_r(\vr)\setminus\mathcal{F}_l(\vr)\}\quad\text{by theorem}~\ref{FT6}\\
&=&\sel(A)\setminus\ser(A)\subseteq\se(A)\subseteq\sigma_S(A).
\end{eqnarray*}
Similarly
\begin{eqnarray*}
\sigma_{-\infty}^S(A)&=&\{\qu\in\quat~|~R_\qu(A)\in\mathcal{SF}(\vr)\quad\text{and}\quad \ind(R_\qu(A))=-\infty\}\\
&=&\{\qu\in\quat~|~R_\qu(A)\in\mathcal{SF}(\vr)\quad\text{and}\quad \dim(\kr(R_{\oqu}(A^\dagger)))=+\infty\}\\
&=&\{\qu\in\quat~|~R_\qu(A)\in\mathcal{F}_l(\vr)\setminus\mathcal{F}_r(\vr)\}\quad\text{by theorem}~\ref{FT6}\\
&=&\ser(A)\setminus\sel(A)\subseteq\se(A)\subseteq\sigma_S(A).
\end{eqnarray*}
Therefore we get
$$\sigma_{+\infty}^S(A)\cup\sigma_{-\infty}^S(A)=\{\qu\in\sigma_S(A)~~|~~R_\qu(A)\in\mathcal{SF}(\vr)\setminus\FF(\vr)\}.$$
\end{proof}
%%%%%%%%%%%%%%%%%%%%%%%%%%%%%%%%%%%%%%
\begin{proposition}\label{Pes}
Let $A\in\B(\vr)$, then we have
$$\se(A)=\left(\sel(A)\cap\ser(A)\right)\cup\sigma_{+\infty}^S(A)\cup\sigma_{-\infty}^S(A)$$
with
$\left(\sel(A)\cap\ser(A)\right)\cap\left(\sigma_{+\infty}^S(A)\cup\sigma_{-\infty}^S(A)\right)=\emptyset.$
\end{proposition}
\begin{proof}
Since $\se(A)=\sel(A)\cup\ser(A)\subseteq\sigma_S(A)$, we get by proposition \ref{EP1},
$$\sel(A)\cap\ser(A)=\{\qu\in\quat~|~R_\qu(A)\not\in\mathcal{SF}(\vr)\}\subseteq\se(A),$$
and  by proposition \ref{Pin} we get
$$\sigma_{+\infty}^S(A)\cup\sigma_{-\infty}^S(A)=\{\qu\in\sigma_S(A)~~|~~R_\qu(A)\in\mathcal{SF}(\vr)\setminus\FF(\vr)\}.$$
Therefore
$$\se(A)=\left(\sel(A)\cap\ser(A)\right)\cup\sigma_{+\infty}^S(A)\cup\sigma_{-\infty}^S(A),$$
and $\left(\sel(A)\cap\ser(A)\right)\cap\left(\sigma_{+\infty}^S(A)\cup\sigma_{-\infty}^S(A)\right)=\emptyset.$
\end{proof}
%%%%%%%%%%%%%%%%%%%%%%%%%%%%%%%%
\begin{proposition}\label{PK1}
Let $\overline{\Z}=\Z\cup\{+\infty,-\infty\}$. For $A\in\B(\vr)$, $K\in\B_0(\vr)$, and $k\in\overline{\Z}\setminus\{0\}$ we have $\sigma_k^S(A+K)=\sigma_k^S(A)$.
\end{proposition}
\begin{proof}Let
 $A\in\B(\vr)$, $K\in\B_0(\vr)$, then
 $R_\qu(A+K)=R_\qu(A)+K_1$, where $K_1=AK+KA-2\text{Re}(\qu)K$ and, by proposition \ref{PP2}, $K_1\in\B_0(\vr)$. Therefore, for $k\in\Z\setminus\{0\}$, by theorem \ref{FT4}, $R_\qu(A)\in\FF(\vr)$ implies $R_\qu(A+K)\in\FF(\vr)$ and $\ind(R_\qu(A+K))=\ind(R_\qu(A))$. Thus, for $k\in\Z\setminus\{0\}$,
 $$\sigma_k^S(A+K)=\sigma_k^S(A).$$
 Now by propositions \ref{EP5} and \ref{Pin}, we have
 $$\sigma_{+\infty}^S(A+K)=\sigma_{+\infty}^S(A)\quad\text{and}\quad \sigma_{-\infty}^S(A+K)=\sigma_{-\infty}^S(A).$$
\end{proof}
%%%%%%%%%%%%%%%%%%%%%%%%%%%%%%
\begin{remark}\label{RK}
Let $A\in\B(\vr)$. The results of proposition \ref{PK1} not true for $\sigma_0^S(A)$. Since $\sigma_0^S(A)=\{\qu\in\sigma_S(A)~|~R_\qu(A)\in\Wy(\vr)\}$ and, according to theorem \ref{FT6},
\begin{eqnarray*}
& &\FF(\vr)=\FF_l(\vr)\cap\FF_r(\vr)\\
&=&\{A\in\B(\vr)~|~\ra(R_\qu(A))~~\text{ closed,}~~\dim(\kr(R_\qu(A)))<\infty~~~\text{and}~~\dim(\kr(R_{\oqu}(A^\dagger)))<\infty\}
\end{eqnarray*}
we can write 
$$\sigma_0^S(A)=\{\qu\in\sigma_S(A)~|~\ra(R_\qu(A))~~\text{ closed,}~~\dim(\kr(R_\qu(A)))=\dim(\kr(R_{\oqu}(A^\dagger)))<\infty\}.$$
Since $\sigma_{pS}(A)=\{\qu\in\quat~|~\kr(R_\qu(A))\not=\{0\}\}$, we have
$$\sigma_0^S(A)=\{\qu\in\sigma_S(A)~|~\ra(R_\qu(A))=\overline{\ra(R_\qu(A))}\not=\vr, \dim(\kr(R_\qu(A)))=\dim(\kr(R_{\oqu}(A^\dagger)))<\infty\}.$$
Therefore, if $\dim(\vr)<\infty$, we have
$$\sigma_0^S(A)=\sigma_{pS}(A)=\sigma_S(A).$$
Suppose that  $\dim(\vr)<\infty$, then $\B(\vr)=\B_0(\vr)$. Since
$$
\Iop^2-2\text{Re}(\qu)\Iop+|\qu|^2\Iop=(1-2\text{Re}(\qu)+|\qu|^2)\Iop
=(1-\qu)(1-\oqu)\Iop,$$
$R_\qu(\Iop)$ is invertible if and only if $\qu\not=1$ and $\oqu\not=1$. That is, $R_\qu(\Iop)$ is invertible if and only if $\qu\not=1$. Thus $\sigma_0^S(\Iop)=\sigma_S(\Iop)=\{1\}.$ Also $R_\qu(\Iop-\Iop)=R_\qu(0)=|\qu|^2\Iop$ is invertible if and only if $\qu\not=0$, thus $\sigma_0^S(0)=\sigma_S(0)=\{0\}.$ That is,
$$\sigma_0^S(\Iop+K)\not=\sigma_0^S(\Iop)$$
with $K=-\Iop$, a compact operator on $\vr$.
\end{remark}	
%%%%%%%%%%%%%%%%%%%%%%%%%%%%%%%%%
\section{The Weyl S-spectrum on $\vr$}
In this section we define the S-Weyl spectrum on $\vr$ and give a characterization to the S-spectrum in terms of the Weyl spectrum.
\begin{definition}\label{WD1}
The S-Weyl spectrum of an operator $A\in\B(\vr)$ is the set
$$\ws(A)=\bigcap_{K\in\B_0(\vr)}\sigma_S(A+K).$$
Hence, by the definition, $\ws(A)$ is the largest part of $\sigma_S(A)$ such that $\ws(A+K)=\ws(A)$ for every $K\in\B_0(\vr)$. Since the S-Weyl spectrum is the intersection of compact sets in $\quat$, $\ws(A)$ is a compact subset of $\quat$.
\end{definition}
%%%%%%%%%%%%
\begin{definition}\label{WD2}
Let $\iso(A)$ denotes the set of all isolated points of the S-spectrum $\sigma_S(A)$, that is
$$\iso(A)=\{\qu\in\sigma_S(A)~~|~~\qu~~\text{is an isolated point of}~~\sigma_S(A)\}.$$
Its compliment in $\sigma_S(A)$,  $\acc(A)=\sigma_S(A)\setminus\iso(A)$, is the set of all accumulation points. Also we denote $\pi_0(A)=\iso(A)\cap \sigma_0^S(A)$.
\end{definition}
\begin{remark}
By proposition \ref{EP}, the isolated eigenvalues are in fact isolated spheres in $\quat$. However we denote the sphere $[\qu]$ by $\qu$.
\end{remark}
%%%%%%%%%%%%%%%%%
\begin{lemma}\label{WL1}
Let $A\in\B(\vr)$. If $A\in\mathcal{SF}(\vr)$ with $\ind(A)\leq 0$ then there is a compact (in fact a finite rank) operator $K\in\B_0(\vr)$ such that $\kr(A+K)=\{0\}$.
\end{lemma}
\begin{proof}
Let $A\in\mathcal{SF}(\vr)$. If $\ind(A)\leq 0$, then $\dim(\kr(A))\leq\dim(\kr(A^\dagger))$, and by  theorem \ref{FT6}, $\dim(\kr(A))<\infty$. Let $\{\phi_i\}_{i=1}^n$ be an orthonormal basis for $\kr(A)$ and let $B$ be an orthonormal basis for $\kr(A^\dagger)=\ra(A)^\perp$, where the cardinality of $B$, $|B|\geq n$. Let $\{\psi_k\}_{k=1}^n\subseteq B$ be an orthonormal set. Define the map $K:\vr\longrightarrow\vr$ by
$$K\phi=\sum_{j=1}^n\psi_j\langle \phi_j|\phi\rangle\quad \text{for each}~~\phi\in\vr,$$
which is clearly right linear. Since
$$\ra(K)\subseteq\text{right}-\quat-\text{span}\{\psi_j\}_{j=1}^n\subseteq\text{right}-\quat-\text{span}~ B=\kr(A^\dagger)=\ra(A)^\perp,$$
$K$ is bounded and finite rank, hence compact. Let $\psi\in\kr(A)$, then by proposition \ref{P1},
$$\|\psi\|^2=\sum_{i=1}^n|\langle\phi_i|\psi\rangle|^2=\|K\psi\|^2.$$
Now, if $\psi\in\kr(A+K)$, then $A\psi=-K\psi$, and therefore,
$$A\psi\in\ra(A)\cap\ra(K)\subseteq\ra(A)\cap\ra(A)^\perp=\{0\}.$$
Thus $A\psi=0$ and $\|\psi\|=\|K\psi\|=\|A\psi\|=0$, which implies $\psi=0$. Hence $\kr(A+K)=\{0\}$.
\end{proof}
%%%%%%%%%%%%%%%%%%%%%%%%%%%%%%%%
\begin{proposition}\label{WP1}
Let $A\in\mathcal{SF}(\vr)$, then $\ind(A)=0$ if and only if there exist a compact operator (in fact, finite rank) operator $K\in\B_0(\vr)$ such that $A+K$ is invertible.
\end{proposition}
\begin{proof}
If $A\in\mathcal{SF}(\vr)$ with $\ind(A)=0$, then $A\in\FF(\vr)$ and , by lemma \ref{WL1}, there exist a compact (in fact, finite rank) operator $K\in\B_0(\vr)$ such that $\kr(A+K)=\{0\}$. Now, by theorem \ref{FT4}, $A+K\in\FF(\vr)$ and $\ind(A+K)=\ind(A)=0$. Thus $\kr((A+K)^\dagger)=\{0\}$ as $\kr(A+K)=\{0\}$. Therefore, $\kr((A+K)^\dagger)=\ra(A+K)^\perp=\{0\}$, which means $\ra(A+K)=\vr$. Therefore $A+K$ is invertible.\\
Conversely, if there exists $K\in\B_0(\vr)$ such that $A+K$ is invertible, then by remark \ref{ER2} (f), $A+K$ is Weyl. Since $A=(A+K)-K$, by theorem \ref{FT4}, $A\in\FF(\vr)$ and $\ind(A)=0$.
\end{proof}
%%%%%%%%%%%%%%%%%%%%%%%%%%%%%
The following theorem characterizes the S-spectrum in terms of the Weyl operators, see definition \ref{ED2} and proposition \ref{WP2}.
\begin{theorem}\label{WT1}(Schechter Theorem) If $A\in\B(\vr)$, then
$$\ws(A)=\se(A)\cup\bigcup_{k\in\Z\setminus\{0\}}\sk(A)=\sigma_S(A)\setminus\sigma_0^S(A).$$
\end{theorem}
\begin{proof}
Let $A\in\B(\vr)$.\\
{\em Claim:} If $\qu\in\sigma_0^S(A)$, then there is a $K\in\B_0(\vr)$ such that $\qu\not\in\sigma_S(A+K)$.\\
For, if $\qu\in\sigma_0^S(A)$, then $R_\qu(A)\in\mathcal{SF}(\vr)$ with $\ind(R_\qu(A))=0$. Since, for $K\in\B_0(\vr)$, $R_\qu(A+K)=R_\qu(A)+K_1$ with, by proposition \ref{PP2}, $K_1=AK+KA-2\text{Re}(\qu)K\in\B_0(\vr)$, by proposition \ref{WP1}, $R_\qu(A+K)$ is invertible. Therefore, $\qu\in\rho_S(A+K)$ or $\qu\not\in\sigma_S(A+K)$ as claimed.\\
Now by proposition \ref{EP7},
$$\sigma_S(A)=\se(A)\cup\bigcup_{k\in\Z}\sk(A),$$
where all the above sets are pairwise disjoint, therefore,
$$\sigma_0^S(A)=\sigma_S(A)\setminus\left(\se(A)\cup\bigcup_{k\in\Z\setminus\{0\}}\sk(A)\right).$$
Let $\qu\in\sigma_S(A)$, if $\qu\in\se(A)\cup\bigcup_{k\in\Z\setminus\{0\}}\sk(A)$, then by propositions \ref{EP5} and \ref{PK1}, for every $K\in\B_0(\vr)$,
$\qu\in\se(A+K)\cup\bigcup_{k\in\Z\setminus\{0\}}\sk(A+K)$. On the other hand, if $\qu\in\sigma_0^S(A)$, then by the above claim $\qu\not\in\sigma_S(A+K)$. Therefore, since the S-Weyl spectrum $\ws(A)$ is the largest part of the S-spectrum $\sigma_S(A)$ that remains invariant under compact perturbations, we have
$$\ws(A)=\se(A)\cup\bigcup_{k\in\Z\setminus\{0\}}\sk(A)=\sigma_S(A)\setminus\sigma_0^S(A).$$
\end{proof}
%%%%%%%%%%%%%%%%%%%%%%%%%%%
\begin{remark}\label{WR1}
Let $A\in\B(\vr)$. From the above theorem, theorem \ref{WT1}, we can  make the following straight forward observations:
\begin{enumerate}
\item [(a)] $\se(A)\subseteq \ws(A)\subseteq\sigma_S(A)$.
\item[(b)]$\displaystyle\se(A)=\ws(A)\Longleftrightarrow \bigcup_{k\in\Z\setminus\{0\}}\sk(A)=\emptyset.$
\item[(c)]$\sigma_S(A)=\ws(A)\cup\sigma_0^S(A)$ and $\ws(A)\cap\sigma_0^S(A)=\emptyset$.
\item[(d)]$\ws(A)=\sigma_S(A)\Longleftrightarrow\sigma_0^S(A)=\emptyset$.
\item[(e)] $\displaystyle\se(A)=\ws(A)=\sigma_S(A)\Longleftrightarrow\bigcup_{k\in\Z}\sk(A)=\emptyset.$
\end{enumerate}
\end{remark}
%%%%%%%%%%%%%%%%%%%%%%%%%
\begin{proposition}\label{WP2}
For every $A\in\B(\vr)$, $\ws(A)=\{\qu\in\quat~~|~~R_\qu(A)\in\B(\vr)\setminus\Wy(\vr)\}.$
\end{proposition}
\begin{proof}
If $\qu\in\rho_S(A)$, then $R_\qu(A)$ is invertible. Since, by remark \ref{ER2} (f), invertible operators are Weyl, $R_\qu(A)\in\Wy(\vr)$. Therefore, if $R_\qu(A)\not\in\Wy(\vr)$, then $\qu\in\sigma_S(A)$. Also by definition \ref{ED2} and theorem \ref{WT1} we have $\ws(A)=\sigma_S(A)\setminus\sigma_0^S(A)$ and $\sigma_0^S(A)=\{\qu\in\sigma_S(A)~~|~~R_\qu(A)\in\Wy(\vr)\}$. Therefore, $\ws(A)=\{\qu\in\quat~~|~~R_\qu(A)\in\B(\vr)\setminus\Wy(\vr)\}$.
\end{proof}
%%%%%%%%%%%%%%%%%%%%%
\begin{remark}\label{WR2}
By remark \ref{ER2} (b), $A\in\Wy(\vr)$ if and only if $A^\dagger\in\Wy(\vr)$. Also by proposition \ref{WP2}, $\qu\in\ws(A)$ if and only if $\oqu\in\ws(A^\dagger)$. Therefore. $\ws(A)=\ws(A^\dagger)^*$.
\end{remark}
%%%%%%%%%%%%%%%
\begin{proposition}\label{WP3} 
$$\Wy(\vr)=\{A\in\B(\vr)~~|~~0\in\rho_S(A)\cup\sigma_0^S(A)\}=\{A\in\FF(\vr)~~|~~~0\in\rho_S(A)\cup\sigma_0^S(A)\}.$$
\end{proposition}
\begin{proof}
Let $A\in\B(\vr)$. If $A\in\Wy(\vr)$, then $A\in\FF(\vr)$ and, by remark \ref{ER2} (e), $A+K\in\Wy(\vr)$ for some $K\in\B_0(\vr)$. Therefore, by proposition \ref{WP1}, $A+K$ is invertible and hence $R_0(A+K)=(A+K)^2$ is invertible. Thus $0\in\rho_S(A+K)$, which means $0\not\in\sigma_S(A+K)$. Therefore, by the definition of the Weyl S-spectrum $0\not\in\ws(A)$, and hence by theorem \ref{WT1}, $0\in \rho_S(A)\cup\sigma_0^S(A)$ . Conversely, let $0\in\rho_S(A)\cup\sigma_0^S(A)$. If $0\in\rho_S(A)$, then $R_0(A)=A^2$ is invertible and hence $A$ is invertible. Therefore, by remark \ref{ER2} (f), $A\in\Wy(\vr)$. If $0\in\sigma_0^S(A)$, then by definition \ref{ED2}, $R_0(A)=A^2\in\Wy(\vr)$ and $0\in\sigma_S(A)$. Thus by theorem \ref{FT2}, $A\in\Wy(\vr)$. Hence, $\Wy(\vr)=\{A\in\B(\vr)~~|~~0\in\rho_S(A)\cup\sigma_0^S(A)\}$, and since $\Wy(\vr)\subseteq\FF(\vr)$, we get $\Wy(\vr)=\{A\in\FF(\vr)~~|~~~0\in\rho_S(A)\cup\sigma_0^S(A)\}.$
\end{proof}
%%%%%%%%%%%%%%%%%%%%%%%%%%%%
\begin{remark}
\begin{enumerate}
\item [(a)]Let $A\in\B(\vr)$, by theorem \ref{WT1}, $\se(A)\subseteq\ws(A)$. By proposition \ref{EP4}, $\se(A)\not=\emptyset$ if and only if $\dim(\vr)=\infty$. Hence, $\ws(A)=\emptyset$ implies $\dim(\vr)<\infty$. Further, since $\ws(A)=\{\qu\in\quat~~|~~R_\qu(A)\in\B(\vr)\setminus\Wy(\vr)\}$, by remark \ref{ER2} (g), $\dim(\vr)<\infty$ implies $\ws(A)=\emptyset$. Therefore, $\ws(A)\not=\emptyset$ if and only if $\dim(\vr)=\infty$.
\item[(b)]Let $K\in\B_0(\vr)$ and $\qu\not=0$. Since $R_\qu(K)=K^2-2\text{Re}(\qu)K+|\qu|^2\Iop$, clearly $|\qu|^2\Iop$ is Fredholm with $\ind(|\qu|^2\Iop)=0$ and $K^2-2\text{Re}(\qu)K$ is compact, by theorem \ref{FT4}, $R_\qu(K)$ is Fredholm with $\ind(R_\qu(K))=0$. That is, $R_\qu(K)\in\Wy(\vr)$. Therefore, by proposition \ref{WP2}, $\ws(K)\setminus\{0\}=\emptyset$. Thus, if $\dim(\vr)=\infty$, then by item (a) $\ws(K)=\{0\}$. Since, by proposition \ref{EP4}, $\se(K)\subseteq\ws(K)$, if $\dim(\vr)=\infty$, then $\emptyset\not=\se(K)\subseteq\{0\}$. Hence, for $K\in\B_0(\vr)$,  $\dim(\vr)=\infty$ if and only if $\se(K)=\ws(K)=\{0\}$.
\item[(c)] Suppose that $A\in\B(\vr)$ is normal and Fredholm. Then clearly $R_\qu(A)$ is normal. Therefore, by proposition \ref{KR}, $R_\qu(A)\in\Wy(\vr)$. Thus, by corollary \ref{EC2} and proposition \ref{WP2}, $T$ is normal implies $\se(A)=\ws(A)$.

\end{enumerate}
\end{remark}
%%%%%%%%%%%%%%%%%%%%%%%%%%%%%%%%%%
\section{Browder S-spectrum in $\vr$}
In the complex case, the Browder theory uses the so-called  Riesz points and  Riesz idempotent and which is defined in terms of Cauchy integral formula. The Cauchy integral formula, in the quaternionic setting, is only available on an axially symmetric domain for slice regular functions in a quaternion slice. Due to this we are unable to provide a complete study on the Browder S-spectrum. However, in this section, we provide certain results about Browder operator and Browder S-spectrum on the whole set of quaternions without the use of the  Riesz idempotent.\\

Let $A\in\B(\vr)$ and $\N_0$ be the set of all non-negative integers. From now on we denote $\kr(A)=\cK(A)$ and $\ra(A)=\cR(A)$. Then, for $n\in\N_0$, clearly
$$\cK(A^n)\subseteq\cK(A^{n+1})\quad\text{and}\quad \cR(A^{n+1})\subseteq\cR(A^n),$$
which means, in the inclusion ordering, $\{\cK(A^n)\}$ and $\{\cR(A^n)\}$ are nondecreasing and non-increasing sequences of $\vr$ respectively.
%%%%%%%%%%%%%
\begin{lemma}\label{BL1}
Let $n_0\in\N_0$.
\begin{enumerate}
\item [(a)]If $\cK(A^{n_0+1})=\cK(A^{n_0})$, then $\cK(A^{n+1})=\cK(A^n)$ for every $n\geq n_0$.
\item[(b)]If $\cR(A^{n_0+1})=\cR(A^{n_0})$, then $\cR(A^{n+1})=\cR(A^n)$ for every $n\geq n_0$.
\end{enumerate}
\end{lemma}
\begin{proof}
The proof is same as the complex proof. For a complex proof see lemma 5.29 in \cite{kub}.
\end{proof}
\begin{definition}\label{BD1}
Let $\overline{\N}_0=\N_0\cup\{+\infty\}$. The ascent and descent of an operator $A\in\B(\vr)$ are defined respectively as follows.
$$\asc(A)=\min\{n\in\overline{\N}_0~~|~~\cK(A^{n+1})=\cK(A^n)\},$$
$$\dsc(A)=\min\{n\in\overline{\N}_0~~|~~\cR(A^{n+1})=\cR(A^n)\}.$$
\end{definition}
Note that for $A\in\B(\vr)$ clearly we have
$$\asc(A)=0\Leftrightarrow \cK(A)=\{0\},\quad\text{that is},~~A~~\text{is injective}.$$
$$\dsc(A)=0\Leftrightarrow \cR(A)=\vr,\quad\text{that is},~~A~~\text{is surjective}.$$
%%%%%%%%%%%%%%%%%
\begin{lemma}\label{BL2}
Let $A\in\B(\vr)$.
\begin{enumerate}
\item [(a)] If $\asc(A)<\infty$ and $\dsc(A)=0$, then $\asc(A)=0$.
\item[(b)] If $\asc(A)<\infty$ and $\dsc<\infty$, then $\asc(A)=\dsc(A).$
\end{enumerate}
\end{lemma}
\begin{proof}
The proof is exactly same as its complex counterpart. For a complex proof see lemma 5.30 in \cite{kub}.
\end{proof}
%%%%%%%%%%%%%%%%%%%
\begin{lemma}\label{BL3}
If $A\in\FF(\vr)$, then $asc(A)=dsc(A^\dagger)$ and $dsc(A)=asc(A^\dagger)$.
\end{lemma}
\begin{proof}
By corollary \ref{CN}, if $A\in\FF(\vr)$ then $A^n\in\FF(\vr)$ for all $n\in\N_0$. Therefore, by proposition \ref{FP1}, $\cR(A^n)$ is closed for every $n\in\N_0$. With these facts the proof follows from its complex version. For a complex proof see lemma 5.31 in \cite{kub}.
\end{proof}
%%%%%%%%%%%%%%%%%%%%
\begin{definition}\label{BD2}
A right quaternionic Browder operator is a right quaternionic Fredholm operator with finite ascent and finite descent. Let $\mathfrak{Br}(\vr)$ denotes the set of all right quaternionic Browder operators from $\B(\vr)$. Then
$$\Br(\vr)=\{A\in\FF(\vr)~~|~~\asc(A)<\infty\quad\text{and}\quad\dsc(A)<\infty\}.$$
\end{definition}
Note that according to lemma \ref{BL2},
\begin{eqnarray}\label{WE1}
\Br(\vr)&=&\{A\in\FF(\vr)~~|~~\asc(A)=\dsc(A)<\infty\}\\
&=&\{\{A\in\FF(\vr)~~|~~\asc(A)=\dsc(A)=m\quad\text{for some}~~m\in\N_0\}\nonumber
\end{eqnarray}
Hence
\begin{equation}\label{WE2}
\FF(\vr)\setminus\Br(\vr)=\{A\in\FF(\vr)~~|~~\asc(A)=\infty~~\text{or}~~\dsc(A)=\infty\}
\end{equation}
Also by lemma \ref{BL3}
\begin{equation}\label{BE3}
A\in\Br(\vr)\Leftrightarrow A^\dagger\in\Br(\vr).
\end{equation}
%%%%%%%%%%%%%
\begin{definition}\label{BD3}
Let $\mathcal{X}$ be a right quaternionic linear space. The subspaces $\cR$ and $\cK$ are said to be algebraic compliments of each other if
$$\mathcal{X}=\cR+\cK\quad\text{and}\quad \cR\cap\cK=\{0\}.$$
\end{definition}
%%%%%%%%%%%%%%%%%%%%
\begin{lemma}\label{BL4}
If $A\in\B(\vr)$ with $\asc(A)=\dsc(A)=m$ for some $m\in\N_0$, then $\cR(A^m)$ and $\cK(A^m)$ are algebraic compliments of each other.
\end{lemma}
\begin{proof}
The proof is purely algebraic and it is the same as its complex counterpart. For a complex proof see lemma 5.32 in \cite{kub}.
\end{proof}
%%%%%%%%%%%%%%%%%%%
\begin{theorem}\label{BT1}
Let $A\in\B(\vr)$. Consider
\begin{enumerate}
\item [(a)] $A\in\Br(\vr)$ and $A\not\in\mathcal{G}(\vr)$.
\item[(b)] $A\in\FF(\vr)$ is such that $\cR(A^m)$ and $\cK(A^m)$ are complimentary subspaces for some $m\in\N$.
\item[(c)]$A\in\Wy(\vr)$.
\end{enumerate}
Then (a)$\Longrightarrow$(b)$\Longrightarrow$(c).
\end{theorem}
\begin{proof}
(a)$\Rightarrow$(b): $A\in\B(\vr)$ is invertible if and only if $A\in\GG(\vr)$. That is, $A\in\B(\vr)$ has a bounded inverse and $\cR(A)=\vr$ , $\cK(A)=\{0\}$. By corollary \ref{CI1}, every invertible operator is Fredholm. Also by definition \ref{BD1}, $\cR(A)=\vr$ and $\cK(A)=\{0\}$ if and only if $\asc(A)=\dsc(A)=0$. Thus every invertible operator is Browder. That is, clearly the inclusion is strict, $\GG(\vr)\subset\Br(\vr)\subset\FF(\vr)$. Therefore, if $A\in\Br(\vr)$ and $A\not\in\GG(\vr)$, then $\asc(A)=\dsc(A)=m$ for some $m\geq 1$. Hence by lemma \ref{BL4}, $\cK(A^m)$ and $\cR(A^m)$ are complimentary subspaces of $\vr$.\\
(b)$\Rightarrow$ (c): Suppose (b) holds. That is, $A\in\FF(\vr)$ and there exists $m\in\N$ such that $\cR(A^m)+\cK(A^m)=\vr$ and $\cR(A^m)\cap\cK(A^m)=\{0\}$. Since $A\in\FF(\vr)$, by corollary \ref{CN}, $A^m\in\FF(\vr)$. Hence $\cK(A^m)$, $\cK((A^m)^\dagger)$ are finite dimensional, and by proposition \ref{FP1}, $\cR(A^m)$ is closed. Since $\cR(A^m)$ is closed, $\cR(A^m)+\cR(A^m)^\perp=\vr$, where $\cR(A^m)^\perp=\cK((A^m)^\dagger)$. Hence, $\cR(A^m)+\cK((A^m)^\perp)=\vr$ and $\cR(A^m)\cap\cK((A^m)^\dagger)=\{0\}$. Thus $\cK(A^m)$ and $\cK((A^m)^\dagger)$ are both algebraic compliments of $\cR(A^m)$, and therefore, they have the same finite dimension. Thus, by remark \ref{FR1}, $\ind(A^m)=0$. Since by corollary \ref{CN}, $m\ind(A)=\ind(A^m)=0$ and $m>0$, we have $ind(A)=0$. Hence $A\in\Wy(\vr)$.
\end{proof}
%%%%%%%%%%%%%%%
\begin{remark}\label{BR1}
From theorem \ref{BT1}, it is now clear that, obviously the inclusions are strict, $\GG(\vr)\subset\Br(\vr)\subset\Wy(\vr)\subset\FF(\vr)$.
\end{remark}
%%%%%%%%%%%%%%%%%%%%
\begin{theorem}\label{BT2}
Let $A\in\B(\vr)$. If $A\in\Br(\vr)$ and $0\in\sigma_S(A)$, then $0\in\pi_0(A)$.
\end{theorem}
\begin{proof}
Since $\sigma_0^S(A)=\{\qu\in\quat~~|~~R_\qu(A)\in\Wy(\vr)\}$ and $\Br(\vr)\subset\Wy(\vr)$, if $A\in\Br(\vr)$ and $0\in\sigma_S(A)$, then $0\in\sigma_0^S(A)$. Further, by theorem \ref{BT1}, there is an integer $m\geq 1$ such that $\cR(A^m)+\cK(A^m)=\vr$ and $\cR(A^m)\cap\cK(A^m)=\{0\}$, equivalently (the direct sum is not necessarily orthogonal)  $\cR(A^m)\oplus\cK(A^m)=\vr$ and $\cR(A^m)\cap\cK(A^m)=\{0\}$. Further, $\cR(A^m)$ and $\cK(A^m)$ are $A^m$-invariant. Therefore, $A^m=A^m\vert_{\cR(A^m)}\oplus A^m\vert_{\cK(A^m)}$. Since $A^m$ is not invertible, $\cK(A^m)\not=\{0\}$ and $A^m\vert_{\cK(A^m)}=O$, the zero operator. Thus by the spectral mapping theorem (see theorem 4.3 (d) in \cite{ghimorper}), $\sigma_S(A)^m=\sigma_S(A^m)=\sigma_S(A^m\vert_{\cR(A^m)})\cup\{0\}$. Since $A^m\vert_{\cR(A^m)}:\cR(A^m)\longrightarrow\cR(A^m)$ is bijective and, by corollary \ref{CN}, $A^m\in\FF(\vr)$, hence by proposition \ref{FP1}, $\cR(A^m)$ is a closed subspace of $\vr$, $A^m\vert_{\cR(A^m)}\in\GG(\cR(A^m))$. Therefore $0\in\rho(A^m\vert_{\cR(A^m)})$, thus $0\not\in\sigma_S(A^m\vert_{\cR(A^m)})$. Hence $\sigma_S(A^m)$ is a disconnected set, and therefore $0$ is an isolated point of $\sigma_S(A)^m=\sigma_S(A^m)$. Thus $0$ is an isolated point of $\sigma_S(A)$. That is, $0\in\iso(A)$ and by definition of $\pi_0(A)$, $0\in\pi_0(A)=\sigma_0^S(A)\cap\iso(A)$.
\end{proof}
%%%%%%%%%%%%%%%%%%%%%%%%%
\begin{remark}\label{BR2}Let $A\in\B(\vr)$.
\begin{eqnarray*}
\sigma_S(A)\setminus\pi_0(A)&=&\sigma_S(A)\setminus(\iso(A)\cap\sigma_0^S(A))\\
&=&\left(\sigma_S(A)\setminus\iso(A)\right)\cup\left(\sigma_S(A)\setminus\sigma_0^S(A)\right)\\
&=&\acc(A)\cup\ws(A)\quad\text{by theorem \ref{WT1}}.
\end{eqnarray*}
\end{remark}
%%%%%%%%%%%%%%%%%%%
\begin{definition}\label{BD4}
The Browder S-spectrum of an operator $A\in\B(\vr)$, denoted by $\Bs(A)$, is
$$\Bs(A)=\{\qu\in\quat~~|~~R_\qu(A)\not\in\Br(\vr)\}.$$
\end{definition}
By equation \ref{BE3}, $A\in\Br(\vr)$ if and only if $A^\dagger\in\Br(\vr)$. Hence
\begin{equation}\label{BE4}
\Bs(A)=\Bs(A^\dagger)^*
\end{equation}
Also by proposition \ref{WP2}, 
\begin{equation}\label{BE5}
\ws(A)\subseteq\Bs(A).
\end{equation}
%%%%%%%%%%%%%%%%%%%
\begin{proposition}\label{BP1}
Let $A\in\B(\vr)$, then $\Bs(A)\subseteq\sigma_S(A)$.
\end{proposition}
\begin{proof}
If $\qu\in\Bs(A)$ then $R_\qu(A)\not\in\Br(\vr)$, so either $R_\qu(A)\not\in\FF(\vr)$ or $R_\qu(A)\in\FF(\vr)$ and $\asc(R_\qu(A))=dsc(R_\qu(A))=\infty$. If $A\not\in\FF(\vr)$, then by corollary \ref{EC2}, $\qu\in\se(A)\subseteq\sigma_S(A)$. If $R_\qu(A)\in\FF(\vr)$ and $\asc(R_\qu(A))=\dsc(R_\qu(A))=\infty$, then, since $\cK(R_\qu(A))=\{0\}$ if and only if $\asc(R_\qu(A))=0$, $R_\qu(A)$ is not invertible. Therefore $\qu\not\in\rho_S(A)$ and hence $\qu\in\sigma_S(A)$. Thus $\Bs(A)\subseteq\sigma_S(A)$.
\end{proof}
%%%%%%%%%%%%%%%%%%%%%%%%
\begin{remark}\label{BR3}
By remark \ref{WR1}, equation \ref{BE5} and proposition \ref{BP1}, for $A\in\B(\vr)$, we have
$$\se(A)\subseteq\ws(A)\subseteq\Bs(A)\subseteq\sigma_S(A).$$
\end{remark}
%%%%%%%%%%%%%%%%%%%%%%
\section{conclusion}
We have studied Weyl operators and Weyl S-spectrum of a bounded quaternionic right linear operator. We have also given a characterization to the S-spectrum in terms of Weyl operators. We have studied the Browder operators and the Browder S-spectrum in a limited sense, which is due to the unavailability of the Cauchy integral formula on the whole set of quaternions. However, Using the Cauchy integral formula and the S-functional calculus on an axially symmetric domain for slice-regular functions, which is accessible \cite{ACS,ACK,GP}, one may define the  Riesz idempotent and study the Browder spectrum on axially symmetric slice domains in the point of view of the S-spectrum. However, we have avoided studying it in this manuscript and we will treat it elsewhere.

%%%%%%%%%%%%%%%%%%%%%%%%%%%%%%%%%%%%%%%%%%%%%%%%%

\end{document}